\documentclass{research4cacm}
\usepackage{upgreek}
\usepackage[
  pdftitle={Abstracting Abstract Machines: A Systematic Approach to Higher-Order Program Analysis},
  pdfauthor={David Van Horn and Matthew Might},
  pdfsubject={Computer Science},
  pdfkeywords={abstract machines, abstract interpretation},
  breaklinks=true
]
{hyperref}
\newcommand{\ie}{\emph{i.e.}}
\newcommand{\etal}{\emph{et al.}}

\newcommand{\syn}[1]{\mathit{#1}} 
\newcommand{\var}[1]{\mathit{#1}}
\newcommand{\s}[1]{\mathit{#1}}

\newcommand{\parto}{\rightarrow_{\mathrm{fin}}}
\newcommand{\dom}{\var{dom}}

\newcommand{\set}[1]{\left\{#1\right\}}

\newcommand{\Pow}[1]{{\mathcal{P}\left(#1\right)}}

\newcommand{\wt}{\sqsubseteq}

\newcommand{\join}{\sqcup}

\newcommand{\bigjoin}{\bigsqcup}

\newenvironment{grammar}{\begin{array}{r@{\;}c@{\;}l@{\;}c@{\;}l@{\;\;\;\;}l}}{\end{array}}

\newcommand{\opor}{\mathrel{|}}

\newcommand{\produces}{=}

\newcommand{\vv}{x}

\newcommand{\schfalse}{{\mbox{\tt\#f}}}

\newcommand{\ttlp}{\mbox{\tt (}}
\newcommand{\ttrp}{\mbox{\tt )}}

\newcommand{\appform}[2]{\ttlp #1 #2\ttrp}
\newcommand{\lamform}[2]{\ttlp \uplambda #1.#2\ttrp}

\newcommand{\ifform}[3]{\ttlp {\tt if}\; #1\; #2\; #3\ttrp}

\newcommand{\expr}{e}

\newcommand{\State}{\Sigma}
\newcommand{\state}{\varsigma}

\newcommand{\Env}{\s{Env}}

\newcommand{\Den}{\s{Val}} 

\newcommand{\store}{\sigma}

\newcommand{\env}{\rho}

\newcommand{\cont}{\kappa}

\newcommand{\alloc}{\mathit{alloc}}

\newcommand{\addr}{a}

\newcommand{\tm}{t}
\newcommand{\tick}{{{tick}}}

\newcommand{\sa}[1]{\widehat{\mathit{#1}}}

\newcommand{\aState}{{\hat{\Sigma}}}
\newcommand{\astate}{{\hat{\varsigma}}}

\newcommand{\astore}{{\hat{\sigma}}}

\newcommand{\aaddr}{{\hat{\addr}}}
\newcommand{\aalloc}{{\widehat{alloc}}}

\newcommand{\atick}{{\widehat{tick}}}

\newcommand{\absmap}{\alpha}

\newcommand\secstar{\texorpdfstring{\textsuperscript{\normalsize $\boldsymbol{*}$}}{*}}

\newcommand\CEK{{\mathit{CEK}}}
\newcommand\ACEK{{\widehat{\CEK}}}

\newcommand\CESKP{{\mathit{CESK}^*}}
\newcommand\CESKPT{{\mathit{CESK}^*_t}}
\newcommand\ACESKPT{{\widehat{\mathit{CESK}^*_t}}}

\newcommand\CESP{{CESK$^*$}}

\newcommand{\multistep}{\longmapsto\!\!\!\!\!\rightarrow}

\newcommand\tmnext{u}

\newcommand\addrnext{b}
\newcommand\addrnextnext{c}
\newcommand\Storable{Storable}
\newcommand\sto{s}
\newcommand\den{v}

\newcommand\mtk{\mathbf{mt}}
\newcommand\fnk{\mathbf{fn}}
\newcommand\ark{\mathbf{ar}}

\newcommand\mte{\emptyset} 

\newcommand\real{\mathit{real}}

\newcommand\liveloc{\mathit{LL}}

\newcommand\betavalue{{\mathrel{\mathbf{v}}}}

\newcommand\evf{\mathit{eval}} 
\newcommand\inj{\mathit{inj}} 
\newcommand\fv{\mathbf{fv}}

\newcommand\restrict[2]{#1 | #2}

\newtheorem{theorem}{Theorem}
\newtheorem{lemma}{Lemma}

\begin{document}
\CopyrightYear{2011}

\subtitle{A Systematic Approach to Higher-Order Program Analysis}
\title{Abstracting Abstract Machines
\thanks{The original version of this paper was published as 
  ``Abstracting Abstract Machines'' in
  \emph{Proceedings of the 15th ACM SIGPLAN International Conference
    on Functional Programming}.}  }

\numberofauthors{2}
\author{
\alignauthor
David {Van Horn}\titlenote{Supported by the National Science
  Foundation under grant 0937060 to the Computing Research Association
  for the CIFellow Project.}\\
       \affaddr{Northeastern University}\\
       \affaddr{Boston, Massachusetts}\\
       \email{dvanhorn@ccs.neu.edu}
\alignauthor
Matthew Might\\
       \affaddr{University of Utah}\\
       \affaddr{Salt Lake City, Utah}\\
       \email{might@cs.utah.edu}
}

\maketitle 

\begin{abstract} 
  Predictive models are fundamental to engineering reliable software systems.
  However, designing conservative, computable approximations for the behavior
  of programs (static analyses) 
  remains a difficult and error-prone process
  for modern high-level programming languages.
  What analysis designers need is a principled method for navigating the gap
  between semantics and analytic models:
  analysis designers need a method that tames the \emph{interaction} of
  complex languages features such as higher-order functions, recursion,
  exceptions, continuations, objects and dynamic allocation.  

  We contribute a \emph{systematic approach to program
    analysis} that yields novel and transparently sound static
  analyses.  Our approach relies on existing derivational techniques
  to transform high-level language semantics into low-level
  deterministic state-transition systems (with potentially infinite
  state spaces).  We then perform a series of simple machine
  refactorings to obtain a sound, computable approximation, which
  takes the form of a \emph{non-deterministic} state-transition
  systems with \emph{finite} state spaces. The approach scales up
  uniformly to enable program analysis of realistic language features,
  including higher-order functions, tail calls, conditionals, side
  effects, exceptions, first-class continuations, and even garbage
  collection.
\end{abstract}

\section{Introduction}

Software engineering, compiler optimizations, program parallelization,
system verification, and security assurance depend on program
analysis, a ubiquitous and central theme of programming language
research.  At the same time, the production of modern software systems
employs expressive, higher-order languages such as Java, JavaScript,
C\#, Python, Ruby, etc., implying a growing need for fast, precise,
and scalable higher-order program analyses.

Program analysis aims to soundly predict properties of
programs before being run. 
  (\emph{Sound} in program analysis means ``conservative approximation'':
  if a sound analysis says a program must not exhibit behavior, then
    that program will not exhibit that behavior;
    but if a sound analysis says a program may exhibit a behavior,
    then it may or may not exhibit that behavior.)
For over thirty years, the research community has expended significant
effort designing effective analyses for higher-order
programs~\cite{dvanhorn:Midtgaard2010Controlflow}.
Past approaches have focused on connecting high-level language
semantics such as structured operational semantics, denotational
semantics, or reduction semantics to equally high-level but dissimilar
analytic models.  These models are too often far removed from their
programming language counterparts and take the form of constraint
languages specified as relations on sets of program
fragments~\cite{dvanhorn:wright-jagannathan-toplas98,dvanhorn:Neilson:1999,dvanhorn:Meunier2006Modular}.
These approaches require significant ingenuity in their design and
involve complex constructions and correctness arguments, making it
difficult to establish soundness, design algorithms, or grow the
language under analysis.  Moreover, such analytic models, which focus
on ``value flow'', \ie,~determining which syntactic values may show
up at which program sites at run-time, have a limited capacity to
reason about many low-level intensional properties such as memory
management, stack behavior, or trace-based properties of computation.
Consequently, higher-order program analysis has had limited impact on
large-scale systems, despite the apparent potential for program
analysis to aid in the construction of reliable and efficient
software.

In this paper, we describe a \emph{systematic approach to program
  analysis} that overcomes many of these limitations by providing a
straightforward derivation process, lowering verification costs and
accommodating sophisticated language features and program properties.

Our approach relies on leveraging existing techniques to transform
high-level language semantics into \emph{abstract
  machines}---low-level deterministic state-transition systems with
potentially infinite state spaces.  Abstract
machines~\cite{dvanhorn:landin-64}, and the paths from semantics to
machines~\cite{dvanhorn:reynolds-acm72,dvanhorn:Danvy:DSc,dvanhorn:Felleisen2009Semantics},
have a long history in the research on programming languages.

From an abstract machine, which represents the idealized core of a
realistic run-time system, we perform a series of basic machine
refactorings to obtain a \emph{non-deterministic} state-transition
system with a \emph{finite} state space.  The refactorings are simple:
(1) variable bindings and the control stack are redirected through the
machine's store and (2) the store is bounded to a finite size.  Due to
finiteness, store updates must become merges, leading to the
possibility of multiple values residing in a single store location.
This in turn requires store look-ups be replaced by a non-deterministic
choice among the multiple values at a given location.  The derived
machine computes a sound approximation of the original machine, and
thus forms an \emph{abstract interpretation} of the machine and the
high-level semantics.

The approach scales up uniformly to enable program analysis of
realistic language features, including higher-order functions, tail
calls, conditionals, side effects, exceptions, first-class
continuations, and even garbage collection.
Thus, we are able to refashion semantic techniques used to model
language features into abstract interpretation techniques for
reasoning about the behavior of those very same features.

\emph{Background and notation}: We present a brief introduction to
reduction semantics and abstract machines.  For background and a more
extensive introduction to the concepts, terminology, and notation
employed in this paper, we refer the reader to \emph{Semantics
  Engineering with PLT Redex}~\cite{dvanhorn:Felleisen2009Semantics}.

\section{From Semantics to Machines and Machines to Analyses}
\label{sec:cek-to-acesk}

In this section, we demonstrate our systematic approach to analysis by
stepping through a derivation from the high-level semantics of a
prototypical higher-order programming language to a low-level abstract
machine, and from the abstract machine to a sound and computable
analytic model that predicts intensional properties of that machine.
As a prototypical language, we choose the call-by-value
$\lambda$-calculus \cite{dvanhorn:Plotkin1975Callbyname}, a core
computational model for both functional and object-oriented languages.
We choose to model program behavior with a simple operational model
given in the form of a reduction semantics.  Despite this simplicity,
reduction semantics scale to full-fledged programming
languages~\cite{dvanhorn:Sperber2010Revised}, although the choice is
somewhat arbitrary since it is known how to construct abstract
machines from a number of semantic
paradigms~\cite{dvanhorn:Danvy:DSc}.
In subsequent sections, we demonstrate the approach handles richer
language features such as control, state, and garbage collection, and
we have successfully employed the same method to statically reason
about language features such as laziness, exceptions, and
stack-inspection, and programming languages such as Java and
JavaScript.
In all cases, analyses are derived following the systematic approach
presented here.

\subsection{Reduction semantics}
\label{sec:red}

To begin, consider the following language of expressions:%
\[
\begin{grammar}
  \expr &\in& \syn{Exp} &\produces& 
 \vv \opor \appform{\expr}{\expr} \opor   \lamform{\vv}{\expr}
\\
  \vv &\in& \syn{Var} & &\text{ an infinite set of identifiers}
  \text.
\end{grammar}
\]
The syntax of expressions includes variables, applications, and
functions.  Values $\den$, for the purposes of this language, include
only function terms, $\lamform\vv\expr$.  We say $\vv$ is
the \emph{formal parameter} of the function $\lamform\vv\expr$,
and $\expr$ is its \emph{body}.
A \emph{program} is a closed expression, \ie,~an expression in which
every variable occurs within some function that binds that variable as
its formal parameter.
Call-by-value \emph{reduction} is characterized by the relation $\betavalue$:
\[
\begin{array}{rcl}
\appform{\lamform{\vv}{\expr}}{\den} & \betavalue & [\den/\vv]\expr\text,
\end{array}
\]
which states that a function applied to a value reduces to the body of
the function with every occurrence of the formal parameter replaced by
the value.  The expression on the left-hand side is a known as
a \emph{redex} and the right-hand side is its \emph{contractum}.

Reduction can occur within a context of an \emph{evaluation context},
defined by the following grammar:
\[
\begin{grammar}
 && E &\produces& [\;] \opor \appform{E}{\expr} \opor \appform{\den}{E}\text.
\end{grammar}
\]
An evaluation context can be thought of as an expression with a single
``hole'' in it, which is where a redex may be reduced.  It is
straightforward to observe that for all programs, either the program
is a value, or it decomposes uniquely into an evaluation context and
redex, written $E[\appform{\lamform{\vv}{\expr}}{\den}]$.
Thus the grammar as given specifies a deterministic reduction
strategy, which is formalized as a \emph{standard reduction relation}
on programs:
\[
E[\expr] \longmapsto_\betavalue E[\expr'],\mbox{ if } \expr\;\betavalue\;\expr'\text.
\]
The \emph{evaluation} of a program is defined by a partial function
relating programs to values~\cite[page
67]{dvanhorn:Felleisen2009Semantics}:
\[
\evf(\expr) = \den\text{ if } \expr \multistep_\betavalue \den\text{, for some }\den\text,
\]
where $\multistep_\betavalue$ denotes the reflexive, transitive
closure of the standard reduction relation.

We have now established the high-level semantic basis for our
prototypical language.  The semantics is in the form of an evaluation
function defined by the reflexive, transitive closure of the
standard reduction relation.
However, the evaluation function as given does not shed much light on
a realistic implementation.  At each step, the program is traversed
according to the grammar of evaluation contexts until a redex is
found.  When found, the redex is reduced and the contractum is plugged
back into the context.  The process is then repeated, again traversing
from the beginning of the program.
Abstract machines offer an extensionally equivalent but more realistic
model of evaluation that short-cuts the plugging of a contractum back
into a context and the subsequent
decomposition~\cite{dvanhorn:Danvy-Nielsen:RS-04-26}.

\subsection{CEK machine}
\label{sec:cek}

The CEK
machine~\cite[Interpreter~III]{dvanhorn:reynolds-acm72}\cite[page
100]{dvanhorn:Felleisen2009Semantics} is a state transition system
that efficiently performs evaluation of a program.  There are
two key ideas in its construction, which can be carried out
systematically~\cite{dvanhorn:Biernacka2007Concrete}.  The first is
that substitution, which is not a viable implementation strategy, is
instead represented in a delayed, explicit manner as
an \emph{environment} structure.  So a substitution $[\den/\vv]\expr$
is represented by $\expr$ and an environment that maps $\vv$ to
$\den$.  Since $\expr$ and $\den$ may have previous substitutions
applied, this will likewise be represented with environments.  So in
general, if $\env$ is the environment of $\expr$ and $\env'$ is the
environment of $\den$, then we represent $[\den/\vv]\expr$ by $\expr$
in the environment $\env$ extended with a mapping of $\vv$ to
$(\den,\env')$, written $\env[\vv\mapsto(\den,\env')]$.  The pairing
of a value and an environment is known as
a \emph{closure}~\cite{dvanhorn:landin-64}.

The second key idea is that evaluation contexts are constructed
inside-out and represent continuations:
\begin{enumerate}
\setlength{\itemsep}{1pt}
\item
$[\; ]$ is represented by $\mtk$; 
\item $E[\appform{[\; ]}{\expr}]$ is
represented by $\ark(\expr',\env,\cont)$ where $\env$ closes $\expr'$
to represent $\expr$ and $\cont$ represents $E$;  and
\item $E[\appform{\den}{[\;
]}]$ is represented by $\fnk(\den',\env,\cont)$ where $\env$ closes
$\den'$ to represent $\den$ and $\cont$ represents $E$.
\end{enumerate}
In this way, evaluation contexts form a program stack: $\mathbf{mt}$ is
the empty stack, and $\mathbf{ar}$ and $\mathbf{fn}$ are frames.

States of the CEK machine are triples consisting of an expression, an
 environment that closes the control string, and a continuation:
\[
\begin{grammar}
 \state &\in& \State &=& 
  \syn{Exp} \times \s{Env} \times \s{Cont} 
\\
   \den &\in& \s{Val} &\produces& \lamform{\vv}{\expr}
\\
  \env &\in &\s{Env} &=& \syn{Var} \parto \s{Val} \times \s{Env}
\\
  \kappa &\in& \s{Cont} &\produces&
  \mtk \opor \ark(\expr,\env,\kappa) \opor \fnk(\den,\env,\kappa)\text.
\end{grammar}
\]

\begin{figure}
\[
\begin{array}{@{}l@{\quad}|@{\quad}r@{}}
\multicolumn{2}{c}{\state \longmapsto_\CEK \state'} \\[1mm]
\hline 
\\
\langle\vv, \env, \cont\rangle &
\langle\den,\env', \cont\rangle
 \mbox{ where }\env(\vv) = (\den,\env') 
\\[1mm]
\langle\appform{\expr_0}{\expr_1}, \env, \cont\rangle &
\langle\expr_0, \env, \ark(e_1, \env, \cont)\rangle
\\[1mm]
\langle\den,\env, \ark(e,\env',\cont)\rangle &
\langle\expr,\env',\fnk(\den,\env,\cont)\rangle
\\[1mm]
\langle\den,\env, \fnk(\lamform{\vv}{\expr},\env',\cont)\rangle &
\langle\expr,\env'[\vv \mapsto (\den,\env)], \cont\rangle
\end{array}
\]
\caption{CEK machine.}
\label{fig:cek}
\end{figure}

The transition function for the CEK machine is defined in
Figure~\ref{fig:cek}.
The initial machine state for a program $\expr$ is given by the
$\inj_\CEK$ function:
\[
\inj_\CEK(\expr) = \langle\expr,\mte,\mtk\rangle\text.
\]
Evaluation is defined by the reflexive, transitive closure of the
machine transition relation and a ``$\real$'' function~\cite[page
129]{dvanhorn:Plotkin1975Callbyname} that maps closures to the term
represented:
\[
\evf_\CEK(\expr) = \real(\den,\env)\text{, where }\inj_\CEK(\expr)
\multistep_\betavalue \langle\den,\env,\mtk\rangle\text,
\]
which is equivalent to the $\evf$ function of Section~\ref{sec:red}:
\begin{lemma}[CEK Correctness~\cite{dvanhorn:Felleisen2009Semantics}]
$\evf_{\mathit{CEK}} = \evf$.
\end{lemma}

We have now established a correct low-level evaluator for our
prototypical language that is extensionally equivalent to the
high-level reduction semantics.
However, program analysis is not just concerned with the result of a
computation, but also with \emph{how} it was produced, \ie,~analysis
should predict intensional properties of the machine as it runs a
program.
We therefore adopt a reachable states semantics that relates a
program to the set of all its intermediate steps:
\[
\CEK(\expr) = \{ \state\ |\  \inj_\CEK(\expr) \multistep_{\mathit{CEK}} \state \}
\text.
\]

Membership in the set of reachable states is straightforwardly
undecidable.  The goal of analysis, then, is to construct an
\emph{abstract interpretation}~\cite{dvanhorn:Cousot:1977:AI} 
that is a sound and computable approximation of the $\CEK$ function.

We can do this by constructing a machine that is similar in structure to
the CEK machine:
it is defined by an \emph{abstract state transition} relation,
$\longmapsto_\ACEK$,
which operates over
\emph{abstract states}, $\hat{\State}$, 
that approximate states of the CEK machine.
Abstract evaluation is then defined as:
\[
\ACEK(\expr) = \{ \astate\ |\ \inj_\ACEK(\expr)
\multistep_{\widehat{\mathit{CEK}}} \astate \}
\text.
\]

\begin{enumerate}
\item \emph{Soundness} is achieved by showing 
transitions preserves approximation, so that if
 \(\state \longmapsto_\CEK \state'\) and \(\astate\) approximates
\(\state\),
then there exists an abstract state \(\astate'\) such that
  \(\astate
\longmapsto_\ACEK  \astate'\) and \(\astate'\) approximates \(\state'\).

\item \emph{Decidability} is achieved by constructing the approximation in
such a way that the state-space of the abstracted machine is finite,
which guarantees that for any program $\expr$, the set
$\ACEK(\expr)$ is finite.
\end{enumerate}

\noindent
\textbf{An attempt at approximation}:
A simple approach to abstracting the machine's state space is to apply
a {\em structural abstraction}, which lifts approximation across the
structure of a machine state, \ie,~expressions, environments, and
continuations.
The problem with the structural abstraction approach for the CEK
machine is that both environments and continuations are recursive
structures.
As a result, the abstraction yields objects in an abstract
state-space with recursive structure, implying the space is infinite.

Focusing on recursive structure as the source of the problem, our
course of action is to add a level of indirection, forcing recursive
structure to pass through explicitly allocated addresses.
Doing so unhinges the recursion in the machine's data structures,
enabling structural abstraction via a single point of approximation:
the store.

The next section covers the first of the two steps for refactoring the
CEK machine into its computable approximation: a store component is
introduced to machine states and variable bindings and continuations
are redirected through the store.  This step introduces no
approximation and the constructed machine operates in lock-step with
the CEK machine.  However, the machine is amenable to a direct
structural abstraction.

\subsection{CESK\secstar\ machine}
\label{sec:ceskp}

The states of the \CESP\ machine extend those of the CEK machine to
include a \emph{store}, which provides a level of indirection for
variable bindings and continuations to pass through.
The store is a finite map from \emph{addresses} to \emph{storable
  values}, which includes closures and continuations, and environments
  are changed to map variables to addresses.
When a variable's value is looked-up by the machine, it is now
accomplished by using the environment to look up the variable's
address, which is then used to look up the value.
To bind a variable to a value, a fresh location in the store is
allocated and mapped to the value; the
environment is extended to map the variable to that address.

To untie the recursive structure associated with continuations, we
likewise add a level of indirection through the store and replace the
continuation component of the machine with a \emph{pointer} to a
continuation in the store.
We term the resulting machine the \CESP{} (control, environment,
store, continuation pointer) machine.
\[
\begin{grammar}
  \state &\in& \State &=& 
  \syn{Exp} \times \s{Env} \times \s{Store} \times \s{Addr}
  \\
  \sto &\in& \Storable &=& \Den \times \Env + \s{Cont}
  \\
  \kappa &\in& \s{Cont} &\produces&
  \mtk \opor \ark(\expr,\env,\addr) \opor \fnk(\den,\env,\addr)\text.
\end{grammar}
\]

The transition function for the \CESP\ machine is defined in
Figure~\ref{fig:cesa}.
The initial state for a program is given by the $\inj_\CESKP$
function, which combines the expression with the empty environment
and a store with a single pointer to the empty continuation, whose
address serves as the initial continuation pointer:
\[
\inj_\CESKP(\expr) = \langle\expr,\mte,[\addr_0 \mapsto \mtk],\addr_0\rangle\text.
\]

An evaluation function based on this machine is defined following the
template of the CEK evaluation given in Section~\ref{sec:cek}:
\begin{align*}
\evf_\CESKP(\expr) &= \real(\den,\env,\store)\text{, where }\\
&\inj_\CESKP(\expr) \multistep_\CESKP \langle\den,\env,\store,\addr_0\rangle
\text,
\end{align*}
where the $\real$ function is suitably extended to follow the environment's
indirection through the store.

We also define the set of reachable machine states:
\[
\CESKP(\expr) = \{ \state\ |\  \inj_\CESKP(\expr) \multistep_\CESKP \state \}
\text.
\]

Observe that for any program, the CEK and CESK$^*$ machines
operate in lock-step: each machine transitions, by the corresponding
rule, if and only if the other machine transitions.
\begin{lemma}
$\CESKP(\expr) \simeq \CEK(\expr)$
\end{lemma}
The above lemma implies correctness of the machine.
\begin{lemma}[CESK$^{\boldsymbol{*}}$ Correctness]
\label{lem:pointer-equiv}
$\evf_{\mathit{CESK}^*} = \evf$.
\end{lemma}

\begin{figure}
\[
\begin{array}{@{}l|r@{}}

\multicolumn{2}{c}{\state \longmapsto_{\mathit{CESK}^*} \state'
\text{, where }\cont=\store(\addr),\addrnext \notin \dom(\sigma)} \\[1mm]
\hline & \\

\langle\vv, \env, \store, \addr\rangle &
\langle\den,\env', \store, \addr\rangle
\mbox{ where }(\den,\env') = \store(\env(\vv))
\\[1mm]
\langle\appform{\expr_0}{\expr_1}, \env, \store, \addr\rangle &
\langle\expr_0, \env, \store[\addrnext\mapsto \ark(\expr_1, \env, \addr)], \addrnext\rangle
\\[1mm]
\langle\den,\env,\store, \addr\rangle 
\\
\mbox{ if }\cont = \ark(\expr,\env',\addrnextnext) 
&
\langle\expr,\env',\store[\addrnext\mapsto \fnk(\den,\env,\addrnextnext)],\addrnext\rangle
\\
\mbox{ if } \cont = \fnk(\lamform{\vv}{\expr},\env',\addrnextnext) 
&
\langle\expr,\env'[\vv \mapsto \addrnext], \store[\addrnext \mapsto (\den,\env)], \addrnextnext\rangle
\end{array}
\]
\caption{CESK$^*$ machine.}
\label{fig:cesa}
\end{figure}

\noindent
\textbf{Addresses, abstraction and allocation}:
The CESK$^*$ machine, as defined in Figure~\ref{fig:cesa},
nondeterministically chooses addresses when it allocates a location in
the store, but because machines are identified up to consistent
renaming of addresses, the transition system remains deterministic.

Looking ahead, an easy way to bound the state-space of this machine is
to bound the set of addresses.
But once the store is finite, locations may need to be reused and when
multiple values are to reside in the same location; the store will
have to soundly approximate this by \emph{joining} the values.

In our concrete machine, all that matters about an allocation strategy
is that it picks an unused address.  In the abstracted machine
however, the strategy \emph{will all but certainly have to re-use previously allocated
  addresses}.  The abstract allocation strategy is therefore crucial
to the design of the analysis---it indicates when finite resources
should be doled out and decides when information should deliberately
be lost in the service of computing within bounded resources.  In
essence, the allocation strategy is the heart of an analysis.

For this reason, concrete allocation deserves a bit more attention in
the machine.  An old idea in program analysis is that dynamically
allocated storage can be represented by the state of the computation
at allocation time~\cite[Section 1.2.2]{dvanhorn:Jones1982Flexible,
dvanhorn:Midtgaard2010Controlflow}.  That is, allocation strategies
can be based on a (representation) of the machine history.  
Since machine histories are always fresh, we
we call them \emph{time-stamps}.

A common choice for a time-stamp, popularized by
Shivers~\cite{dvanhorn:Shivers:1991:CFA}, is to represent the history
of the computation as \emph{contours}, finite strings encoding the
calling context.
We present a concrete machine that uses a general time-stamp approach
and is parameterized by a choice of $\tick$ and $\alloc$ functions.

\subsection{Time-stamped CESK\secstar\ machine} 
\label{sec:secpt}

The machine states of the time-stamped \CESP{} machine include a
\emph{time} component, which is intentionally left unspecified:
\[
\begin{array}{r@{\;}c@{\;}l}
  \tm,\tmnext &\in &\s{Time}
  \\
  \state &\in& \State =
  \syn{Exp} \times \s{Env} \times \s{Store} \times \s{Addr} \times \s{Time}  
  \text.
\end{array}
\]
The machine is parameterized by the functions:
\begin{align*}
\tick &: \State \rightarrow \s{Time}
&
\alloc &: \State \to \s{Addr}
\text.
\end{align*}
The $\tick$ function returns the next time; the $\alloc$ function
allocates a fresh address for a binding or continuation.
We require of $\tick$ and $\alloc$ that for all $\tm$ and $\state$,
$\tm \sqsubset \tick(\state)$ and $\alloc(\state) \notin \sigma$ where
$\state = \langle \_,\_,\store,\_, \tm\rangle$.

The time-stamped \CESP{} machine is defined in Figure~\ref{fig:cesat}.
Note that occurrences of $\state$ on the right-hand side of this
definition are implicitly bound to the state occurring on the
left-hand side.
The evaluation function $\evf_\CESKPT$ and reachable states $\CESKPT$
are defined following the same outline as before and omitted for space.
The initial machine state is defined as:
\[
\inj_{\mathit{CESK^*_\tm}}(\expr) = \langle\expr,\mte,[\addr_0 \mapsto \mtk],\addr_0,\tm_0\rangle\text.
\]
\begin{figure}
\[
\begin{array}{@{}l@{\;}|@{\;}r@{}}

\multicolumn{2}{c}{\state \longmapsto_{\mathit{CESK}^*_\tm} \state'
\text{, where }\cont=\store(\addr), \addrnext=\alloc(\state), \tmnext=\tick(\state)} 
\\[1mm]
\hline \\
\langle\vv, \env, \store, \addr,\tm\rangle &
\langle\den,\env', \store, \addr,\tmnext\rangle
\mbox{ where }(\den,\env') = \store(\env(\vv))
\\[1mm]
\langle\appform{\expr_0}{\expr_1}, \env, \store, \addr,\tm\rangle 
&
\langle\expr_0, \env, \store[\addrnext\mapsto \ark(\expr_1, \env, \addr)], \addrnext,\tmnext\rangle
\\[1mm]
\langle\den,\env,\store, \addr,\tm\rangle 
\\
\mbox{if }\cont = \ark(\expr,\env,\addrnextnext) 
&
\langle\expr,\env,\store[\addrnext\mapsto \fnk(\den,\env,\addrnextnext)],\addrnext,\tmnext\rangle
\\
\mbox{if }\cont = \fnk(\lamform{\vv}{\expr},\env',c) 
& 
\langle\expr,\env'[\vv \mapsto b], \store[b \mapsto (\den,\env)], c,\tmnext\rangle
\end{array}
\]
\caption{Time-stamped CESK$^*$ machine.}
\label{fig:cesat}
\end{figure}

Satisfying definitions for the parameters are:
\begin{gather*}
\s{Time} = \s{Addr} = \mathbb{Z}
\\
\begin{align*}
\addr_0 = \tm_0 &= 0
&
\tick\langle \_, \_, \_, \_, \tm\rangle &= \tm+1
&
\alloc \langle \_, \_, \_, \_, \tm\rangle &= \tm
\text.
\end{align*}
\end{gather*}
Under these definitions, the time-stamped CESK$^*$ machine
operates in lock-step with the CESK$^*$ machine, and therefore
with the CEK machine, implying its correctness.
\begin{lemma}
\label{lem:time-stamp-equiv}
$\CESKPT(\expr) \simeq
\CESKP(\expr)$.
\end{lemma}
\noindent
The time-stamped CESK$^*$ machine forms the basis of our
abstracted machine in the following section.

\subsection{Abstract time-stamped CESK\secstar\ machine}
\label{sec:acespt}

As alluded to earlier, with the time-stamped CESK$^*$ machine, we
now have a machine ready for direct abstract interpretation via a
single point of approximation: the store.
Our goal is a machine that resembles the
time-stamped CESK$^*$ machine, but operates over a finite
state-space and it is allowed to be nondeterministic.
Once the state-space is finite, the transitive
closure of the transition relation becomes computable, and this
transitive closure constitutes a static analysis.
Buried in a path through the transitive closure is a possibly
infinite traversal that corresponds to the concrete execution of the
program.

The abstracted variant of the time-stamped CESK$^*$ machine comes from
bounding the address space of the store and the number of times
available.
By bounding the address space, the whole state-space becomes
finite. (Syntactic sets like $\syn{Exp}$ are infinite, but finite for
any given program.)
For the purposes of soundness, an entry in the store may be forced to
hold several values simultaneously:
\[
\begin{grammar}
  \astore &\in &\sa{Store} &=& \s{Addr} \parto \Pow{\Storable}
  \text.
\end{grammar}
\]
Hence, stores now map an address to a \emph{set} of
storable values rather than a single value.  These
collections of values model approximation in the analysis.  If a
location in the store is re-used, the new value is joined with the
current set of values.  When a location is dereferenced, the analysis
must consider any of the values in the set as a result of the
dereference.

The abstract time-stamped CESK$^*$ machine is defined in Figure~\ref{fig:acesat}.
The non-deterministic abstract transition relation changes little compared with
the concrete machine.
We only have to modify it to account for the possibility that multiple
storable values, which includes continuations, may reside together in
the store.  We handle this situation by letting the machine
non-deterministically choose a particular value from the set at a
given store location.

\begin{figure}
\[
\begin{array}{@{}l@{\;}|@{\;}r@{}}

\multicolumn{2}{@{}c@{}}{\astate \longmapsto_{\widehat{\mathit{CESK}^*_\tm}} \astate'
\text{, where } \cont\in\astore(\addr), \addrnext=\aalloc(\astate,\cont), \tmnext=\atick(\astate,\cont)}
\\[2mm]
\hline & \\
\langle\vv, \env, \astore, \addr,\tm\rangle &
\langle \den,\env', \astore, \addr,\tmnext\rangle
\text{ where } (\den,\env') \in \astore(\env(\vv))
\\[1mm]
\langle\appform{\expr_0}{\expr_1}, \env, \astore, \addr,\tm\rangle 
&
\langle\expr_0, \env, \astore \join [\addrnext\mapsto \ark(\expr_1, \env, \addr)], \addrnext,\tmnext\rangle
\\[1mm]
\langle \den,\env,\astore, \addr,\tm\rangle 
\\
\mbox{if }\cont= \ark(\expr,\env',\addrnextnext)
&
\langle\expr,\env',\astore \join [\addrnext\mapsto \fnk(\den,\env,\addrnextnext)],\addrnext,\tmnext\rangle
\\
\mbox{if }\cont=\fnk(\lamform{\vv}{\expr},\env',\addrnextnext)
&
\langle\expr,\env'[\vv \mapsto \addrnext], \astore \join [\addrnext \mapsto (\den,\env)], \addrnextnext,\tmnext\rangle
\end{array}
\]
\caption{Abstract time-stamped CESK$^*$ machine.}
\label{fig:acesat}
\end{figure}

The analysis is parameterized by abstract variants of the functions
that parameterized the concrete version:
\begin{align*}
\atick &: \aState \times \s{Cont} \rightarrow \s{Time}\text,
&
\aalloc &: \aState \times \s{Cont} \to \s{Addr}
\text.
\end{align*}
In the concrete, these parameters determine allocation and
stack behavior.
In the abstract, they are the arbiters of precision: they determine
when an address gets re-allocated, how many addresses get allocated,
and which values have to share addresses.

Recall that in the concrete semantics, these functions consume
states---not states and continuations as they do here.
This is because in the concrete, a state alone suffices since the
state determines the continuation.
But in the abstract, a continuation pointer within a state may denote a
multitude of continuations; however the transition relation is
defined with respect to the choice of a particular one.
We thus pair states with continuations to encode the choice.

The \emph{abstract} semantics is given by the reachable states:
\[
\ACESKPT(\expr) = \{ \astate\ |\ \absmap(\inj_\CESKPT(\expr)) \multistep_{\widehat{\mathit{CESK}^*_\tm}} \astate \}\text.
\]

\noindent
\textbf{Soundness and decidability}:
We have endeavored to evolve the abstract machine gradually so that
its fidelity in soundly simulating the original CEK machine is both
intuitive and obvious.
To formally establish soundness of the abstract time-stamped
\CESP{} machine, we use an abstraction function, defined in
Figure~\ref{fig:abs-map}, from the state-space of the concrete
time-stamped machine into the abstracted state-space.

\begin{figure}
\begin{align*}
  \absmap(\expr,\env,\store,\addr,\tm) &= (e,\absmap(\env),\absmap(\store),\absmap(\addr),\absmap(\tm)) 
    && \text{[states]}
  \\
  \absmap(\env) &= \lambda \vv . \absmap(\env(\vv)) && \text{[environments]}
\\
  \absmap(\store) &= \lambda \aaddr . \!\!\!\! \bigjoin_{\absmap(\addr) = \aaddr} \!\!\!\! \set{\absmap(\store(\addr))} && \text{[stores]}  
\\
  \absmap(\lamform{\vv}{\expr},\env) &= (\lamform{\vv}{\expr},\absmap(\env)) && \text{[closures]}
  \\
  \absmap(\mtk) &= \mtk && \text{[continuations]}
  \\
  \absmap(\ark(\expr,\env,\addr)) &= \ark(\expr,\absmap(\env),\absmap(\addr))
  \\
  \absmap(\fnk(\den,\env,\addr)) &= \fnk(\den,\absmap(\env),\absmap(\addr))
\end{align*}
\caption{Abstraction map, $\absmap :
  \State_{\mathit{CESK}^*_\tm} \rightarrow
  \aState_{\widehat{\mathit{CESK}^*_\tm}}$.}
\label{fig:abs-map}
\end{figure}

The abstraction map over times and addresses is defined so that
the parameters $\aalloc$ and $\atick$ are
sound simulations of the parameters $\alloc$
and $\tick$, respectively.
We also define the partial order $(\wt)$ on the abstract state-space
as the natural point-wise, element-wise, component-wise and
member-wise lifting, wherein the partial orders on the sets
$\syn{Exp}$ and $\s{Addr}$ are flat.
Then, we can prove that abstract machine's transition relation
simulates the concrete machine's transition relation.
\begin{theorem}[Soundness]\ \\ 
  If $\state \longmapsto_{\mathit{CEK}} \state'$ and
  $\alpha(\state) \wt \astate$, then there exists an abstract state
  $\astate'$, such that $\astate
  \longmapsto_{\widehat{\mathit{CESK}}^*_\tm} \astate'$ and
  $\alpha(\state') \wt \astate'$.
\end{theorem}
\begin{proof}
  By Lemmas~\ref{lem:pointer-equiv}
  and~\ref{lem:time-stamp-equiv}, it suffices to prove soundness with
  respect to $\longmapsto_{\mathit{CESK}^*_\tm}$.
  Assume $\state \longmapsto_{\mathit{CESK}^*_\tm} \state'$ and
  $\alpha(\state) \wt \astate$.
  Because $\state$ transitioned, exactly one of the rules from the definition of
  $(\longmapsto_{\mathit{CESK}^*_\tm})$ applies.
  We split by cases on these rules.
  The rule for the second case is deterministic and follows by calculation.
  For the remaining (nondeterministic) cases, we must show
  an abstract state exists such that the simulation is
  preserved.
  By examining the rules for these cases, we see that all three hinge
  on the abstract store in $\astate$ soundly approximating the
  concrete store in $\state$, which follows from the assumption that
  $\absmap(\state) \wt \astate$.
\end{proof}

\begin{theorem}[Decidability]\ \\ 
Membership
of $\astate$ in $\ACESKPT(\expr)$ is decidable.
\end{theorem}
\begin{proof}
  The state-space of the machine is non-recursive with finite sets at the leaves
  on the assumption that addresses are finite.  Hence 
  reachability is decidable since the abstract state-space is finite.
\end{proof}

\section{Abstract state and control}
\label{sec:realistic-features}

We have shown that store-allocated continuations make abstract
interpretation of the CESK$^\star$ machine 
straightforward.
In this section, we want to show that the tight correspondence between
concrete and abstract persists after the addition of language
features such as conditionals, side effects, and first-class
continuations.
We tackle each feature, and present the additional machinery required
to handle each one.
In most cases, the path from a canonical concrete machine to
pointer-refined abstraction of the machine is so simple we only show
the abstracted system.
In doing so, we are arguing that this abstract machine-oriented
approach to abstract interpretation represents a flexible and viable
framework for building program analyses.

To handle conditionals, we extend the language with a new syntactic
form, $\ifform{\expr}{\expr}{\expr}$, and introduce a base value
$\schfalse$, representing false.
Conditional expressions induce a new continuation form:
$\mathbf{if}(\expr'_0,\expr'_1,\env,\addr)$, which represents the
evaluation context $E[\ifform{[\;]}{\expr_0}{\expr_1}]$ where $\env$
closes $\expr'_0$ to represent $\expr_0$, $\env$ closes $\expr'_1$ to
represent $\expr_1$, and $\addr$ is the address of the representation
of $E$.

Side effects are fully amenable to our approach;
we introduce Scheme's \texttt{set!} for mutating variables using the
$\appform{{\tt set!}\;}{\vv\; \expr}$ syntax.  The \texttt{set!}  form
evaluates its subexpression $\expr$ and assigns the value to the
variable $\vv$.  Although \texttt{set!} expressions are evaluated for
effect, we follow Felleisen \etal~and specify \texttt{set!}
expressions evaluate to the value of $\vv$ before it was
mutated~\cite[page 166]{dvanhorn:Felleisen2009Semantics}.  The
evaluation context $E[\appform{{\tt set!}\;}{\vv\; [\;]}]$ is
represented by $\mathbf{set}(\addr_0,\addr_1)$, where $\addr_0$ is the
address of $\vv$'s value and $\addr_1$ is the address of the
representation of $E$.

First-class control is introduced by adding a new base value {\tt
  callcc} which reifies the continuation as a new kind of applicable
value.  Denoted values are extended to include representations of
continuations.  Since continuations are store-allocated, we choose to
represent them by address.  When an address is applied, it represents
the application of a continuation (reified via {\tt callcc}) to a
value.  The continuation at that point is discarded and the applied
address is installed as the continuation.

The resulting grammar is:
\[
\begin{grammar}
 \expr &\in& \syn{Exp} &\produces& \dots \opor \ifform{\expr}{\expr}{\expr}
 \opor \appform{{\tt set!}\;}{\vv\; \expr}
\\
  \kappa &\in& \s{Cont} &\produces&
  \dots \opor \mathbf{if}(\expr,\expr,\env,\addr) \opor \mathbf{set}(\addr,\addr)
\\
  \den &\in& \Den &\produces& \dots \opor \schfalse \opor  {\tt callcc} \opor \addr
  \text.
\end{grammar}
\]
We show only the abstract transitions, which result from
store-allocating continuations, time-stamping, and abstracting the
concrete transitions for conditionals, mutation, and control.
The first three machine transitions deal with conditionals; here we
follow the Scheme tradition of considering all non-false values as
true.
The fourth and fifth transitions deal with mutation.

\begin{figure}
\[
\begin{array}{@{}l@{\;}|@{\;}r@{}}

\multicolumn{2}{@{}c@{}}{\astate \longmapsto_{\widehat{\mathit{CESK}^*_\tm}} \astate'
\text{,\;where\;} 
\cont \in \astore(\addr), 
\addrnext = \aalloc(\astate,\cont),
\tmnext = \atick(\astate,\cont)
 } \\[2mm]
\hline & \\

\langle\ifform{\expr_0}{\expr_1}{\expr_2},\env,\astore,\addr,\tm\rangle
&
\langle\expr_0,\env,\astore \sqcup [\addrnext\mapsto \mathbf{if}(\expr_1,\expr_2,\env,\addr)],\addrnext,\tmnext\rangle
\\[1mm]
\langle\schfalse,\env,\astore,\addr,\tm\rangle
&
\langle\expr_1,\env',\astore,\addrnextnext,\tmnext\rangle
\\
\mbox{if }\cont=
   \mathbf{if}(\expr_0,\expr_1,\env',\addrnextnext)
\\[1mm]
\langle\den,\env,\astore,\addr,\tm\rangle
&
\langle\expr_0,\env',\astore,\addrnextnext,\tmnext\rangle
\\
\mbox{if }\cont=
   \mathbf{if}(\expr_0,\expr_1,\env',\addrnextnext),
\\
\text{and } 
\den\neq\schfalse
\\[1mm]
\langle\appform{{\tt set!}\;}{\vv\; \expr},\env,\astore,\addr,\tm\rangle
&
\langle\expr,\env,\astore \sqcup [\addrnext\mapsto \mathbf{set}(\env(\vv), \addr)],\addrnext,\tmnext\rangle
\\[1mm]
\langle\den,\env,\astore,\addr,\tm\rangle
&
\langle\den',\env,\astore \join [\addr' \mapsto \den],\addrnextnext,\tmnext\rangle
\\
\mbox{ if }\cont=
   \mathbf{set}(\addr',\addrnextnext) 
& \mbox{where }\den' \in \astore(\addr')
\\[1mm]
\langle \lamform{\vv}{\expr},\env,\astore,\addr,\tm\rangle 
&
\langle \expr,\env[\vv\mapsto \addrnext],\astore \join [\addrnext\mapsto \addrnextnext],\addrnextnext,\tmnext\rangle
\\
\mbox{if }\cont= \fnk({\tt callcc},\env',\addrnextnext)
&
\mbox{where }\addrnextnext=\aalloc(\astate,\cont)
\\[1mm]
\langle \addrnextnext,\env,\astore,\addr,\tm\rangle 
&
\langle \addr,\env,\astore,\addrnextnext,\tmnext\rangle 
\\
\mbox{if }\cont= \fnk({\tt callcc},\rho',\addr')
\\[1mm]
\langle \den,\env,\astore,\addr,\tm\rangle 
&
\langle \den,\env,\astore,\addrnextnext,\tmnext\rangle 
\\
\mbox{if }\cont= \fnk(\addrnextnext,\env',\addr')
\end{array}
\]
\caption{Abstract extended CESK$^*$ machine.}
\end{figure}

The remaining three transitions deal with first-class control.
In the first of these, {\tt callcc} is being applied to a closure
value $\den$.  The value $\den$ is then ``called with the current
continuation'', \ie, $\den$ is applied to a value that represents the
continuation at this point.
In the second, {\tt callcc} is being applied to a continuation
(address).  When this value is applied to the reified continuation, it
aborts the current computation, installs itself as the current
continuation, and puts the reified continuation ``in the hole''.
%
Finally, in the third, a continuation is being applied; $\addrnextnext$ gets
thrown away, and $\den$ gets plugged into the continuation $\addrnext$.
In all cases, these transitions result from pointer-refinement,
time-stamping, and abstraction of the usual machine
transitions.

\section{Abstract garbage collection}
\label{sec:garbage-collection}

Garbage collection determines when a store location has become
unreachable and can be re-allocated.
This is significant in the abstract semantics because an address may be
allocated to multiple values due to finiteness of the address space.
Without garbage collection, the values allocated to this common
address must be joined, introducing imprecision in the analysis (and
inducing further, perhaps spurious, computation).
By incorporating garbage collection in the abstract semantics, the
location may be proved to be unreachable and safely \emph{overwritten}
rather than joined, in which case no imprecision is introduced.

Like the rest of the features addressed in this paper, we can
incorporate abstract garbage collection into our static analyzers by a
straightforward pointer-refinement of textbook accounts of concrete
garbage collection, followed by a finite store abstraction.

Concrete garbage collection is defined in terms of a GC machine that
computes the reachable addresses in a store~\cite[page
  172]{dvanhorn:Felleisen2009Semantics}:
\[
\begin{array}{l}
\langle\mathcal{G},\mathcal{B},\store\rangle 
\longmapsto_{\mathit{GC}}
\langle(\mathcal{G}\cup \liveloc_\store(\store(\addr)) \setminus
(\mathcal{B}\cup\{\addr\})), \mathcal{B}\cup\{\addr\}, \store\rangle
\\
\mbox{if }\addr \in \mathcal{G}\text.
\end{array}
\]
This machine iterates over a set of reachable but unvisited ``grey''
locations $\mathcal{G}$.  On each iteration, an element is removed and
added to the set of reachable and visited ``black'' locations
$\mathcal{B}$.  Any newly reachable and unvisited locations, as
determined by the ``live locations'' function $\liveloc_\store$, are
added to the grey set.  When there are no grey locations, the black
set contains all reachable locations.  Everything else is garbage.

The live locations function computes a set of locations which may be
used in the store.  
 Its definition varies based on the machine being garbage collected,
but the definition appropriate for the CESK$^*$ machine of
Section~\ref{sec:ceskp} is:
\begin{align*}
\liveloc_\store(\expr) &= \emptyset
\\
\liveloc_\store(\expr,\env) &= \liveloc_\store(\restrict{\env}{\fv(\expr)})
\\
\liveloc_\store(\env) &= \mathit{rng}(\env)
\\
\liveloc_\store(\mtk) &= \emptyset
\\
\liveloc_\store(\fnk(\den,\env,\addr)) &= \{\addr\} \cup \liveloc_\store(\den,\env) \cup \liveloc_\store(\store(\addr))
\\
\liveloc_\store(\ark(\expr,\env,\addr)) &= \{\addr\}\cup \liveloc_\store(\expr,\env) \cup \liveloc_\store(\store(\addr))\text.
\end{align*}
We write $\restrict\env{\fv(\expr)}$ to
mean $\env$ restricted to the domain of free variables in $\expr$.
We assume the least-fixed-point solution in the calculation of the
function $\liveloc$ in cases where it recurs on itself.

The pointer-refinement requires parameterizing the
$\liveloc$ function with a store used to resolve pointers to
continuations. 
A nice consequence of this parameterization is that we can re-use
$\liveloc$ for \emph{abstract garbage collection} by supplying it an
abstract store for the parameter.
Doing so only necessitates extending $\liveloc$ to the case of sets of
storable values:
\begin{align*}
\liveloc_\store(S) &= \bigcup_{s\in S} \liveloc_\store(s)
\end{align*}

The CESK$^*$ machine incorporates garbage collection by a
transition rule that invokes the GC machine as a subroutine to remove
garbage from the store (Figure~\ref{fig:gc}).
The garbage collection transition introduces non-determinism to the
CESK$^*$ machine because it applies to any machine state and thus
overlaps with the existing transition rules.
The non-determinism is interpreted as leaving the choice of
\emph{when} to collect garbage up to the machine.

The abstract CESK$^*$ incorporates garbage collection by the
\emph{concrete garbage collection transition}, \ie, we re-use the
definition in Figure~\ref{fig:gc} with an abstract store, $\astore$,
in place of the concrete one.
Consequently, it is easy to verify abstract garbage collection
approximates its concrete counterpart.

\begin{figure}
\[
\begin{array}{@{}l@{\qquad}|@{\qquad}r@{}}

\multicolumn{2}{c}{\state \longmapsto_{\mathit{CESK}^*} \state'} \\[1mm]
\hline & \\

\langle \expr,\env,\store,\addr\rangle
\qquad\qquad\qquad
&
\langle \expr,\env,\{\langle\addrnext,\store(\addrnext)\rangle\ |\ \addrnext \in \mathcal{L}\},\addr\rangle
\\
\multicolumn{2}{@{}l}{
\mbox{if }\langle\liveloc_\store(\expr,\env) \cup \liveloc_\store(\store(\addr)),\{\addr\},\store\rangle \multistep_{\mathit{GC}} \langle\emptyset,\mathcal{L},\store\rangle}
\end{array}
\]
\caption{GC transition for the CESK$^*$ machine.}
\label{fig:gc}
\end{figure}

The CESK$^*$ machine may collect garbage at any point in the
computation, thus an abstract interpretation must soundly approximate
\emph{all possible choices} of when to trigger a collection, which 
the abstract CESK$^*$ machine does correctly.
This may be a useful analysis \emph{of} garbage collection, however it
fails to be a useful analysis \emph{with} garbage collection: for
soundness, the abstracted machine must consider the case in which
garbage is never collected, implying no storage is reclaimed to
improve precision.

However, we can leverage abstract garbage collection to reduce the
state-space explored during analysis and to improve precision and
analysis time.
This is achieved (again) by considering properties of the
\emph{concrete} machine, which abstract directly; in this case,
we want the concrete machine to deterministically collect garbage.
Determinism of the CESK$^*$ machine is restored by defining the
transition relation as a non-GC transition (Figure~\ref{fig:cesa})
followed by the GC transition (Figure~\ref{fig:gc}).
This state-space of this concrete machine is ``garbage free'' and
consequently the state-space of the abstracted machine is ``abstract
garbage free.''

In the concrete semantics, a nice consequence of this property is that
although continuations are allocated in the store, they are
deallocated as soon as they become unreachable, which corresponds to
when they would be popped from the stack in a non-pointer-refined
machine.  Thus the concrete machine really manages continuations like
a stack.

Similarly, in the abstract semantics, continuations are deallocated as
soon as they become unreachable, which often corresponds to when they
would be popped.  We say often, because due to the finiteness of the
store, this correspondence cannot always hold.
However, this approach gives a good finite approximation to infinitary
stack analyses that can always match calls and returns.

\section{Related work}

The study of abstract machines for the $\lambda$-calculus began with
Landin's SECD machine~\cite{dvanhorn:landin-64}, the systematic
construction of machines from semantics with Reynolds's definitional
interpreters~\cite{dvanhorn:reynolds-acm72}, the theory of abstract
interpretation with the seminal work of Cousot and
Cousot~\cite{dvanhorn:Cousot:1977:AI}, and static analysis of the
$\lambda$-calculus with Jones's coupling of abstract machines and
abstract interpretation~\cite{dvanhorn:Jones:1981:LambdaFlow}.
All have been active areas of research since their inception,
but only recently have well known abstract machines been connected
with abstract interpretation by Midtgaard and
Jensen~\cite{dvanhorn:midtgaard-jensen-sas-08,dvanhorn:Midtgaard2009Controlflow}.
We strengthen the connection by demonstrating a general technique for
abstracting abstract machines.

The approximation of abstract machine states for the analysis of
higher-order languages goes back to
Jones~\cite{dvanhorn:Jones:1981:LambdaFlow},
who argued abstractions of regular tree automata could solve the
problem of recursive structure in environments.
We re-invoked that wisdom to eliminate the recursive structure of
continuations by allocating them in the store.

Midtgaard and Jensen present a 0CFA for a CPS
language~\cite{dvanhorn:midtgaard-jensen-sas-08}.  The approach is
based on Cousot-style calculational abstract
interpretation~\cite{dvanhorn:Cousot98-5}, applied to a functional
language.  Like the present work, Midtgaard and Jensen start with a
known abstract machine for the concrete semantics, the CE machine of
Flanagan, \etal~\cite{dvanhorn:Flanagan1993Essence}, and employ a
reachable-states model.  They then compose well-known Galois
connections to reveal a 0CFA with reachability in the style of
Ayers~\cite{dvanhorn:ayers-phd93}.  The CE machine is not sufficient
to interpret direct-style programs, so the analysis is specialized to
programs in continuation-passing style.

Although our approach is not calculational like Midtgaard and
Jensen's, it continues in their vein by applying abstract
interpretation to well known machines, extending the application to
direct-style machines to obtain a parameterized family of analyses
that accounts for polyvariance.

Static analyzers typically hemorrhage precision in the presence of
exceptions and first-class continuations:
they jump to the top of the lattice of approximation when these
features are encountered.
Conversion to continuation- and exception-passing style can handle
these features without forcing a dramatic ascent of the lattice of
approximation~\cite{dvanhorn:Shivers:1991:CFA}.
The cost of this conversion, however, is lost knowledge---both
approaches obscure static knowledge of stack structure, by desugaring
it into syntax.

Might and Shivers introduced the idea of using abstract garbage
collection to improve precision and efficiency in flow
analysis~\cite{dvanhorn:Might:2006:GammaCFA}.
They develop a garbage collecting abstract machine for a CPS language
and prove it correct.  We extend abstract garbage collection to
direct-style languages interpreted on the CESK machine.

\section{Conclusions and perspective}

We have demonstrated a derivational approach to program analysis that
yields novel abstract interpretations of languages with higher-order
functions, control, state, and garbage collection.
These abstract interpreters are obtained by a straightforward pointer
refinement and structural abstraction that bounds the address space,
making the abstract semantics safe and computable.  
The technique allows concrete implementation technology, such as
garbage collection, to be imported straightforwardly into that of
static analysis, bearing immediate benefits.
More generally, an abstract machine based approach to analysis shifts
the focus of engineering efforts from the design of complex analytic
models such as involved constraint languages back to the design of
programming languages and machines, from which analysis can
be \emph{derived}.
Finally, our approach uniformly scales up to richer language features
such as laziness, stack-inspection, exceptions, and object-orientation.
We speculate that store-allocating bindings and continuations is
sufficient for a straightforward abstraction of most existing
machines.

Looking forward, a semantics-based approach opens new possibilities
for design.  Context-sensitive analysis can have daunting
complexity~\cite{dvanhorn:VanHorn-Mairson:ICFP08}, which we have made
efforts to tame~\cite{dvanhorn:Might2010Resolving}, but modular
program analysis is crucial to overcome the significant cost of
precise abstract interpretation.  Modularity can be achieved without
needing to design clever approximations, but rather by
designing \emph{modular semantics} from which modular analyses follow
systematically~\cite{DBLP:journals/corr/abs-1103-1362}. 
Likewise, \emph{push-down analyses} offer infinite state-space
abstractions with perfect call-return matching while retaining
decidability.  Our approach expresses this form of abstraction
naturally: the store remains bounded, but continuations stay on the
stack.

\section{Acknowledgments}

We thank Matthias Felleisen, Jan Midtgaard, Sam Tobin-Hochstadt, and
Mitchell Wand for discussions, and the anonymous reviewers of
\emph{ICFP'10} for their close reading and helpful critiques; their
comments have improved this paper.

\bibliographystyle{abbrv} 


\end{document}